\documentclass{article}

    \PassOptionsToPackage{numbers, compress}{natbib}

    \usepackage[preprint]{neurips_2018}

\usepackage[utf8]{inputenc} 
\usepackage[T1]{fontenc}    
\usepackage{hyperref}       
\usepackage{url}            
\usepackage{booktabs}       
\usepackage{amsfonts}       
\usepackage{nicefrac}       
\usepackage{microtype}      

\usepackage{amssymb}
\usepackage{amsfonts}
\usepackage{amsmath}
\usepackage{amsmath}
\usepackage{amsfonts}
\usepackage{amssymb}
\usepackage{setspace}
\usepackage{enumitem}
\usepackage{graphicx}
\usepackage{longtable}
\usepackage{subcaption} 

\usepackage{float}
\floatstyle{plaintop}
\restylefloat{table}

\usepackage[strict]{changepage}
\usepackage{afterpage}

\usepackage{chngcntr}
\usepackage{apptools}
\usepackage[bottom]{footmisc}
\usepackage{natbib}   
\usepackage{color}
\usepackage{algorithm}
\usepackage{algorithmic}
\usepackage{mathtools}

\newtheorem{theorem}{Theorem}

\newenvironment{proof}[1][Proof]{\noindent \textbf{#1.} }{\  \rule{0.5em}{0.5em}}

\title{Large scale continuous-time mean-variance portfolio allocation via reinforcement learning}

\author{%
  Haoran Wang \\
  Department of Industrial Engineering and Operations Research\\
  Columbia University\\
  New York, NY 10027 \\
}

\begin{document}

\maketitle

\begin{abstract}
We propose to solve large scale Markowitz mean-variance (MV) portfolio allocation problem using reinforcement learning (RL). By adopting the recently developed continuous-time exploratory control framework, we formulate the exploratory MV problem in high dimensions. We further show the optimality of a multivariate Gaussian feedback policy, with time-decaying variance, in trading off exploration and exploitation. Based on a provable policy improvement theorem, we devise a scalable and data-efficient RL algorithm and conduct large scale empirical tests using data from the S\&P $500$ stocks. We found that our method consistently achieves over $10\%$ annualized returns and it outperforms econometric methods and the deep RL method by large margins, for both long and medium terms of investment with monthly and daily trading.
\end{abstract}

\section{Introduction}

Reinforcement learning (RL) has demonstrated to be sucessful in games (\citep{Go1}, \cite{Go2}, \cite{DQN}) and robotics (\cite{robot1}, \cite{robotics2}), which also raised significant attention on its applications to quantitative finance. Notable examples include large scale optimal order execution using clssical Q-learning method (\cite{NFK}),  portfolio allocation using direct policy search (\cite{MS}, \cite{MS2}), and option pricing and hedging using deep RL methods (\cite{deep_hedging}), among others.

However, most existing works only focus on RL problems with expected utility of discounted rewards. Such criteria are either unable to fully characterize the uncertainty of the decision making process in financial  markets or opaque to  typical investors. On the other hand,
mean--variance (MV)  is one of the most important criteria for portfolio choice. Initiated in the Nobel Prize winning work \cite{M} for portfolio selection  in a single period, such a criterion yields an asset allocation strategy that minimizes the variance of the final payoff while targeting some prespecified mean return. The popularity of the MV criterion is not only due to its intuitive and transparent nature  in capturing the tradeoff between  risk and reward for practitioners, but also due to the theoretically interesting  issue of time-inconsistency (or Bellman's inconsistency) inherent with the underlying stochastic optimization and control problems.

In a recent paper \cite{hwang_2}, the authors established an RL framework for studying the continuous-time MV portfolio selection, with continuous portfolio (action) and wealth (state) spaces. Their framework adopts a general entropy-regularized,  relaxed stochastic control formulation, known as the {\it exploratory} formulation, which was originally developed in \cite{Hwang} to capture explicitly
the tradeoff between exploration and exploitation in RL for continuous-time optimization problems. The paper \cite{hwang_2} proved the optimality of Gaussian exploration (with time-decaying variance) for the MV problem in one dimension, and proposed a data-driven algorithm, the EMV algorithm, to learn the optimal Gaussian policy of the exploratory MV problem. Their simulation shows that the EMV algorithm outperforms both a classical econometric method and the deep deterministic policy gradient (DDPG) algorithm by large margins when solving the MV problem in the setting with only one risky asset.

It is the contribution of this work to generalize the continuous-time exploratory MV framework in \cite{hwang_2} to large scale portfolio selection setting, with the number of risky assets being relatively large and the available training data being relatively limited. We establish the theoretical optimality of the high-dimensional Gaussian policy and design a scalable EMV algorithm to directly output portfolio allocation strategies. By switching to portfolio selection in high dimensions, we can in principle take more advantage of the diversification effect (\cite{diversification}) to have better performance while, however, potentially encountering the challenges of low sample efficiency and instability faced by most deep RL methods (\cite{Deep_RL_1}, \cite{Deep_RL_2}). Nevertheless, although the EMV algorithm is an on-policy approach, it can achieve better data efficiency than the off-policy method DDPG, thanks to a provable policy improvement theorem and the explicit functional structures of the theoretical optimal Gaussian policy and value function. For instance, in a $10$ years monthly trading experiment (see section 5.2) where the available data point for training is the same amount as the decision making times for testing, the EMV algorithm still can outperform various alternative methods considered in this paper. To further empirically test the performance and robustness of the EMV algorithm, we conduct experiments using both monthly and daily price data of the S\&P $500$ stocks, for long and medium term investment horizons. Annual returns over $10\%$ have been consistently observed across most experiments. The EMV algorithm also demonstrated remarkable universal applicability, as it can be trained and tested respectively on {\it different} sets of data from stocks that are randomly selected and still achieves competitive and, actually, more robust performance (see Appendix \ref{batch}).

\section{Notations \& Background}
\subsection{Classical continuous-time MV problem}
We consider the classical MV problem in continuous time (without RL), where the investment universe consists of one riskless asset (savings account) and $d$ risky assets (e.g., stocks). Let an investment planning horizon $T>0$ be fixed. Denote by $\{x^u_t,0\leq t\leq T\}$ the {\it discounted} wealth (i.e. state) process of an agent who rebalances her portfolio (i.e. action) investing in  the risky and riskless assets with a strategy (policy) $u=\{u_t, 0\leq t\leq T\}$. Here $u_t=(u_t^1,\dots,u_t^d)$ is the discounted dollar value put in the $d$ risky assets at time $t$. Under the geometric Brownian motion assumption for stocks prices and the standard self-financing condition, it follows (see Appendix \ref{wealth_appendix}) that the wealth process satisfies
\begin{equation}\label{classical_wealth}
dx^u_t=\sigma u_t\cdot (\rho\, dt+\,dW_t), \quad 0\leq t\leq T,
\end{equation}
with an initial endowment being $x^u_0=x_0\in \mathbb{R}$. Here, $W_t=(W^1_t,\dots, W^d_t)$, $0\leq t\leq T$, is a standard $d$-dimensional Brownian motion defined on a filtered probability space $(\Omega, \mathcal{F},\{\mathcal{F}_{t}\}_{0\leq t\leq T},\mathbb{P})$. The vector\footnote{All vectors in this paper are taken as column vectors.} $\rho$ is typically known as the market price of risk, and $\sigma\in \mathbb{R}^{d\times d}$ is the volatility matrix which is assumed to be non-degenerate.

The classical continuous-time MV model then aims to solve the following constrained optimization problem
\begin{eqnarray}
\min_{u} \text{Var}[x^u_T], \quad \text{subject to} \ \mathbb{E}[x^u_T]=z,\label{target}
\end{eqnarray}
where $\{x^u_t,0\leq t\leq T\}$ satisfies the dynamics (\ref{classical_wealth}) under the investment strategy (portfolio) $u$, and $z\in \mathbb{R}$ is an investment target set at $t=0$ as  the desired target payoff at the end of the investment horizon $[0,T]$.

Due to the variance in its objective, (\ref{target}) is known to be {\it time inconsistent}. In this paper we focus ourselves to the so-called \textit{pre-committed} strategies of the MV problem, which are optimal at $t=0$ only. To solve (\ref{target}), one first transforms it into an unconstrained problem by applying a Lagrange multiplier $w$:\footnote{Strictly speaking, $2w\in \mathbb{R}$ is the Lagrange multiplier.}
\begin{equation}\label{unconstrained_classical}
\min_u \mathbb{E}[(x^u_T)^2]-z^2-2w\left(\mathbb{E}[x^u_T]-z\right)=\min_u \mathbb{E}[(x^u_T-w)^2]-(w-z)^2.
\end{equation}
This problem can be solved analytically, whose solution $u^*=\{u^*_t,0\leq t\leq T\}$ depends on $w$.
Then the original constraint $\mathbb{E}[x^{u^*}_T]=z$ determines the value of
$w$.  We refer a detailed derivation to \cite{MV_zhou}.

\subsection{Exploratory continuous-time MV problem}
The classical MV solution requires the estimation of
 the market parameters from historical time series of assets prices. However, as well known in practice, it is difficult to estimate the investment opportunity parameters, especially the mean return vector (aka the {\it mean--blur problem}; see, e.g.,  \cite{Luenberger}) with a workable  accuracy. Moreover, the classical optimal MV strategies are often extremely sensitive to these parameters,  largely due to the procedure of inverting ill-conditioned covariance matrices  to obtain optimal allocation weights. In view of these two issues, the Markowitz solution can be greatly irrelevant to the underlying investment objective.

On the other hand, RL techniques do not require, and indeed often skip, any estimation of model parameters. Rather, RL algorithms, driven by historical data,  output optimal (or near-optimal) allocations directly. This is made possible by
direct interactions with the unknown investment environment, in a learning (exploring) while optimizing (exploiting) fashion.
Following \cite{Hwang}, we introduce the ``exploratory" version of the state dynamics (\ref{classical_wealth}). In this formulation, the control (portfolio) process $u=\{u_t, 0\leq t\leq T\}$ is randomized to represent exploration in RL, leading to
a measure-valued or distributional control process whose density function is given by $\pi=\{\pi_t,0\leq t\leq T\}$. The dynamics (\ref{classical_wealth}) is changed to
\begin{eqnarray}
dX^{\pi}_t &=&\left(\int_{\mathbb{R}^d} \rho' \sigma u\pi_t(u)du\right)\,dt+\left({\int_{\mathbb{R}^d} u'\sigma'\sigma u\pi_t(u)du}\right)^{\frac{1}{2}}\, dB_t,
 \label{state_process0}
\end{eqnarray}
where $\{B_t,0\leq t\leq T\}$ is a 1-dimensional standard Brownian motion on the filtered probability space $(\Omega, \mathcal{F},\{\mathcal{F}\}_{0\leq t\leq T},\mathbb{P})$. Mathematically,
(\ref{state_process0}) coincides with the {\it relaxed control} formulation in classical control theory, and it is adopted here to characterize the effect of exploration on the underlying continuous-time state dynamics change.
We refer the readers to \cite{Hwang} for a  detailed discussion on the motivation of (\ref{state_process0}).

The randomized, distributional  control process $\pi=\{\pi_t, 0\leq t\leq T\}$ is to model
exploration, whose overall level is in turn captured by
its accumulative differential entropy
\begin{equation}\label{entropy}
\mathcal{H}(\pi):=-\int_0^T\int_{\mathbb{R}^d}\pi_t(u)\ln \pi_t(u)dudt.
\end{equation}
Further, we introduce a {\it temperature parameter} (or {\it exploration weight}) $\lambda>0$ reflecting  the tradeoff between exploitation and exploration.
The entropy-regularized, exploratory  MV problem is then to solve, for any fixed $w\in \mathbb{R}$:
\begin{equation}\label{value_function}
\min_{\pi\in \mathcal{A}(0,x_0)}\mathbb{E}\left[(X_T^{\pi}-w)^2+\lambda \int_0^T\int_{\mathbb{R}^d}\pi_t(u)\ln \pi_t(u)dudt\right]-(w-z)^2,
\end{equation}
where $\mathcal{A}(0,x_0)$ is the set of admissible distributional controls on $[0,T]$ whose precise definition is relegated to Appendix \ref{value_function_def}. Once this problem is solved with a minimizer
$\pi^*=\{\pi^*_t, 0\leq t\leq T\}$, the Lagrange multiplier $w$ can  be determined by the additional constraint $\mathbb{E}[X^{\pi^*}_T]=z$.

The optimization objective (\ref{value_function}) explicitly encourages exploration, in contrast to the classical problem (\ref{unconstrained_classical}) which concerns exploitation only.

\section{Optimality of Gaussian Exploration}

To solve the exploratory MV problem (\ref{value_function}), we apply the classical Bellman's principle of optimality for the optimal value function $V$ (see Appendix \ref{value_function_def} for the precise definition of $V$):
$$V(t,x;w)=\inf_{\pi\in \mathcal{A}(t,x)}\mathbb{E}\left[V(s,X^{\pi}_s;w)+\lambda\int_t^s\int_{\mathbb{R}^d}\pi_{l}(u)\ln \pi_{l}(u)dudl\Big | X^{\pi}_t=x\right],$$
for $x\in \mathbb{R}$ and $0\leq t< s\leq T$. Following standard arguments, we deduce that $V$ satisfies the Hamilton-Jacobi-Bellman (HJB) equation
\begin{equation}\label{HJB_2}
v_t(t,x;w)+\min_{\pi\in \mathcal{P}(\mathbb{R}^d)}\int_{\mathbb{R}^d}\left(\frac{1}{2}u'\sigma'\sigma u v_{xx}(t,x;w)+\rho'\sigma u v_x(t,x;w)+\lambda \ln\pi(u)\right)\pi(u)du=0,
\end{equation}
with the terminal condition $v(T,x;w)=(x-w)^2-(w-z)^2$. Here, $\mathcal{P}%
\left( \mathbb{R}^d\right) $ denotes the set of density functions of probability measures on $\mathbb{R}^d$ that are absolutely
continuous with respect to the Lebesgue measure and $v$ denotes the generic unknown solution to the HJB equation.

Applying the usual verification technique and using the fact that $\pi\in \mathcal{P}(\mathbb{R}^d)$ if and only if
$\int_{\mathbb{R}^d}\pi (u)du=1$ and $\pi (u)\geq 0$, a.e., on $\mathbb{R}^d$,
we can solve the (constrained) optimization problem in the HJB equation (\ref{HJB_2}) to obtain a feedback  (distributional) control whose density function is given by
\begin{eqnarray}\nonumber
\boldsymbol{\pi} ^{\ast }(u;t,x,w)&=&\frac{\exp\left(-\frac{1}{\lambda}\left(\frac{1}{2}u'\sigma'\sigma u v_{xx}(t,x;w)+\rho' \sigma u v_x(t,x;w)\right)\right)}{\int_{\mathbb{R}^d}\exp\left(-\frac{1}{\lambda}\left(\frac{1}{2}u'\sigma'\sigma u v_{xx}(t,x;w)+\rho' \sigma u v_x(t,x;w)\right)\right)du}\\
&=& \mathcal{N}\left(\, u\, \Big | -{\sigma}^{-1}\rho\frac{v_x(t,x;w)}{v_{xx}(t,x;w)}\ , \ \left(\sigma'\sigma\right)^{-1}\frac{\lambda}{v_{xx}(t,x;w)}\right),\label{Gaussian}
\end{eqnarray}
where $\mathcal{N}(u|\beta,\Sigma)$ denotes the Gaussian density function with mean vector $\beta$ and covariance matrix $\Sigma$. It is assumed in (\ref{Gaussian}) that $v_{xx}(t,x;w)>0$, which will be verified in what follows.

Substituting the candidate optimal Gaussian feedback control policy (\ref{Gaussian}) back into the HJB equation (\ref{HJB_2}), the latter is transformed to
\begin{equation}\label{HJB_3}
v_t(t,x;w)-\frac{\rho'\rho}{2}\frac{v^2_x(t,x;w)}{v_{xx}(t,x,w)}+\frac{\lambda}{2}\left(d-d\ln\left( \frac{2\pi e \lambda}{v_{xx}(t,x;w)}\right)+\ln\left(|\sigma'\sigma|\right)\right)=0,
\end{equation}
with $v(T,x;w)=(x-w)^2-(w-z)^2$, where $|\cdot|$ denotes the matrix determinant. A direct computation yields that this equation has a classical solution
\begin{equation}\label{solution_to_HJB}
v(t,x;w)=(x-w)^2e^{-\rho'\rho(T-t)}+\frac{\lambda d}{4}\rho'\rho\left(T^2-t^2\right)-\frac{\lambda d}{2}\left(\rho'\rho T-\frac{1}{d}\ln \frac{|\sigma'\sigma|}{\pi \lambda}\right)(T-t)-(w-z)^2,
\end{equation}
which clearly satisfies  $v_{xx}(t,x;w)>0$, for any $(t,x)\in [0,T]\times \mathbb{R}$. It then follows that the candidate optimal feedback Gaussian policy (\ref{Gaussian}) reduces to
\begin{equation}\label{Gaussian_explicit}
\boldsymbol{\pi} ^{\ast }(u;t,x,w)=\mathcal{N}\left(\, u\, \Big|  -{\sigma}^{-1}\rho(x-w)\ , \  \left(\sigma'\sigma\right)^{-1}\frac{\lambda}{2}e^{\rho'\rho (T-t)}\right),\;\;
(t,x)\in [0,T]\times \mathbb{R}.
\end{equation}

Finally, the  optimal wealth process (\ref{state_process0}) under $\boldsymbol{\pi} ^{\ast }$ becomes
\begin{equation}\label{wealth_SDE}
dX^{*}_t=-\rho'\rho (X^{*}_t-w)\, dt+\left(\rho'\rho\left(X^{*}_t-w\right)^2+\frac{\lambda}{2}e^{\rho'\rho(T-t)}\right)^{\frac{1}{2}}\, dB_t,\;
X^{*}_0=x_0.
\end{equation}
It has a unique strong solution for $0\leq t\leq T$, as can be easily verified. We now summarize the above results in the following theorem.
\begin{theorem}\label{verification}
The optimal value function of the entropy-regularized exploratory MV problem (\ref{value_function}) is given by
\begin{equation}\label{value_verified}
V(t,x;w)=(x-w)^2e^{-\rho'\rho(T-t)}+\frac{\lambda d}{4}\rho'\rho\left(T^2-t^2\right)-\frac{\lambda d}{2}\left(\rho'\rho T-\frac{1}{d}\ln \frac{|\sigma'\sigma|}{\pi \lambda}\right)(T-t)-(w-z)^2,
\end{equation}%
for $(t,x)\in [0,T]\times \mathbb{R}$.
Moreover,
the optimal feedback control is Gaussian, with its density function given by
\begin{equation}
\boldsymbol{\pi}^{\ast }(u;t,x,w)=\mathcal{N}\left(\, u\, \Big|  -{\sigma}^{-1}\rho(x-w)\ , \  \left(\sigma'\sigma\right)^{-1}\frac{\lambda}{2}e^{\rho'\rho (T-t)}\right).
\label{Gaussian_verified}
\end{equation}%
The associated optimal wealth process under $\boldsymbol{\pi}^{\ast }$
is the unique solution of the stochastic differential equation
\begin{equation}\label{wealth_SDE_1}
dX^{*}_t=-\rho'\rho (X^{*}_t-w)\, dt+\left(\rho'\rho\left(X^{*}_t-w\right)^2+\frac{\lambda}{2}e^{\rho'\rho(T-t)}\right)^{\frac{1}{2}}\, dB_t,\;X^{*}_0=x_0.
\end{equation}
Finally, the Lagrange multiplier $w$ is given by $w=\frac{ze^{\rho'\rho T}-x_0}{e^{\rho'\rho T}-1}$.
\end{theorem}
\begin{proof}
See Appendix \ref{proof1}.
\end{proof}

\smallskip

Theorem \ref{verification} indicates that the level of exploration, measured by the variance of Gaussian policy $\frac{\lambda}{2\sigma^2}e^{\rho^2(T-t)}$, decays in time. The agent initially engages in exploration at the maximum level, and reduces it gradually (although never to zero) as time approaches the end of the investment horizon. Naturally, exploitation dominates exploration as time approaches maturity.  Theorem \ref{verification} presents such a  decaying exploration scheme {\it endogenously} which, to our best knowledge, has not been derived in the RL literature.

Moreover, the mean of the Gaussian distribution (\ref{Gaussian_verified}) 
is independent of the exploration weight $\lambda$, while its variance is independent of the state $x$. This highlights  a \textit{perfect separation} between exploitation and exploration, as the former is captured by the mean and  the latter  by the variance of the optimal Gaussian exploration. This property is also consistent with the linear--quadratic case in the  infinite horizon  studied in \cite{Hwang}.

It is reasonable to expect that the exploratory problem converges to its
classical counterpart as the exploration weight $\lambda$ decreases to 0. Let $\boldsymbol{u}^{\ast}$ be the optimal feedback control for the classical MV problem, and denote by $V^{\text{cl}}$ the optimal value function. Let $\delta_a(\cdot)$ be the Dirac measure centered at $a\in \mathbb{R}^d$. Then the following result holds.

\begin{theorem}\label{convergence_to_Dirac}
For each $(t,x,w)\in [0,T]\times \mathbb{R}\times \mathbb{R}$,
$$\lim_{\lambda \rightarrow 0}\boldsymbol{\pi}^{\ast}(\cdot;t,x;w)=\delta_{\boldsymbol{u}^{\ast}(t,x;w)}(\cdot) \;\;\mbox{ weakly.}$$
Moreover, 
$$\lim_{\lambda \rightarrow 0}|V(t,x;w)-V^{\text{cl}}(t,x;w)|=0.$$

\end{theorem}

\begin{proof}
See Appendix \ref{proof2}.
\end{proof}

\section{RL Algorithm Design}

\subsection{A policy improvement theorem}

We present a policy improvement theorem that is a crucial prerequisite  for our interpretable RL algorithm, the EMV algorithm, which solves the exploratory MV problem in high dimensions.
\begin{theorem}[Policy Improvement Theorem]\label{PIT}
 Let $w\in \mathbb{R}$ be fixed and  $\boldsymbol{\pi}=\boldsymbol{\pi}(\cdot;\cdot,\cdot,w)$ be an arbitrarily given admissible  feedback control policy. Suppose that the corresponding value function $V^{\boldsymbol{\pi}}(\cdot,\cdot;w)\in C^{1,2}([0,T)\times \mathbb{R})\cap C^0([0,T]\times \mathbb{R})$ and satisfies $V^{\boldsymbol{\pi}}_{xx}(t,x;w)>0$, for any $(t,x)\in [0,T)\times \mathbb{R}$. Suppose further that the feedback policy $\tilde{{\boldsymbol{\pi}}}$ defined by
\begin{equation}\label{new_policy}
\tilde{\boldsymbol{\pi}}(u;t,x,w)=\mathcal{N}\left( \, u\, \Big| -\sigma^{-1}\rho\frac{V^{\boldsymbol{\pi}}_x(t,x;w)}{V^{\boldsymbol{\pi}}_{xx}(t,x;w)}\ , \  (\sigma'\sigma)^{-1}\frac{\lambda}{V^{\boldsymbol{\pi}}_{xx}(t,x;w)} \right)
\end{equation}
is admissible. Then, 
\begin{equation}\label{value_improve}
V^{\tilde{{\boldsymbol{\pi}}}}(t,x;w)\leq V^{\boldsymbol{\pi}}(t,x;w),\quad (t,x)\in [0,T]\times \mathbb{R}.
\end{equation}
\end{theorem}
\begin{proof}
See Appendix \ref{proof3}.
\end{proof}

\medskip

The above theorem suggests that there are always policies in the Gaussian family
that improves
the value function of any given, not necessarily Gaussian, policy. Moreover, the Gaussian family is closed under the policy improvement scheme. Hence, without loss of generality, we can simply focus on
the Gaussian policies when choosing an initial solution. The next result shows convergence of both the value functions and the policies from a specifically parameterized Gaussian policy.

\begin{theorem}\label{convergence_learning}
Let $\boldsymbol{\pi}_0(u;t,x,w)=\mathcal{N}(u| \alpha(x-w), \Sigma e^{\beta(T-t)})$, with $\alpha\in \mathbb{R}^d$, $\beta\in \mathbb{R}$ and $\Sigma$ being a $d\times d$ positive definite matrix. Denote by $\{\boldsymbol{\pi}_n(u;t,x,w), (t,x)\in [0,T]\times \mathbb{R},n\geq 1\}$ the sequence of feedback policies updated by the policy improvement scheme (\ref{new_policy}), and $\{V^{\boldsymbol{\pi}_n}(t,x;w), (t,x)\in [0,T]\times \mathbb{R}, n\geq 1\}$ the sequence of the corresponding value functions. Then,
\begin{equation}
\lim_{n\rightarrow \infty} \boldsymbol{\pi}_n (\cdot; t,x,w)= \boldsymbol{\pi^*}(\cdot; t,x,w) \;\;\mbox{ weakly,}
\end{equation}
and
\begin{equation}
\lim_{n\rightarrow \infty} V^{\boldsymbol{\pi}_n}(t,x;w)=V(t,x;w),
\end{equation}
for any $(t,x,w)\in [0,T]\times \mathbb{R}\times \mathbb{R}$, where $\boldsymbol{\pi^*}$ and $V$ are the optimal Gaussian policy (\ref{Gaussian_verified}) and the optimal value function (\ref{value_verified}), respectively.
\end{theorem}
\begin{proof}
See Appendix \ref{proof4}.
\end{proof}

\subsection{The EMV algorithm}
We provide the EMV algorithm to directly learn the optimal solution of the continuous-time exploratory MV problem in high dimensions within competitive training time. Theorem \ref{PIT} provides guidance for policy improvement. For the policy evaluation step, we follow \cite{Doya} to minimize the continuous-time Bellman's error
\begin{equation}
\delta_t:=\dot{V}^{\boldsymbol{\pi}}_t+\lambda\int_{\mathbb{R}^d} \pi_t(u)\ln \pi_t(u)du,
\end{equation}
where $\dot{V}^{\boldsymbol{\pi}}_t=\frac{V^{\boldsymbol{\pi}}(t+\Delta t,X_{t+\Delta t})-V^{\boldsymbol{\pi}}(t,X_t)}{\Delta t}$ is the total derivative and $\Delta t$ is the discretization step for the learning algorithm. This leads to the cost function to be minimized
\begin{equation}\label{learning_objective}
C(\theta,\phi)=\frac{1}{2}\sum_{(t_i,x_i)\in \mathcal{D}}\left(\dot{V}^{\theta}(t_i,x_i)+\lambda\int_{\mathbb{R}^d} \pi^{\phi}_{t_i}(u)\ln \pi^{\phi}_{t_i}(u)du\right)^2 \Delta t,
\end{equation}
using samples collected in the set $\mathcal{D}$ under the current Gaussian policy $\boldsymbol{\pi}^{\phi}$. Here, both the value function $V^{\theta}$ and the Gaussian policy $\pi^{\phi}$ can be  parametrized more explicitly, in view of (\ref{value_verified}), Theorem \ref{PIT} and \ref{convergence_learning}. The cost function (\ref{learning_objective}) can then be minimized by stochastic gradient decent. Finally, the EMV algorithm updates the Lagrange multiplier $w$ every $N$ iterations based on stochastic approximation and the constraint $\mathbb{E}[X^{\pi}_T]=z$, namely, $w\leftarrow w-\alpha(\frac{1}{N}\sum_j x^j_T-z)$, where $x_T^j$'s are the most recent $N$ terminal wealth values. We refer the readers to \cite{hwang_2} for a more detailed description of the EMV algorithm in the one risky asset scenario.

\section{Empirical Results}
\subsection{Data and methods}
We test the EMV algorithm on price data of the S\&P $500$ stocks for both monthly and daily trading. For the former, we train the EMV algorithm on the $10$ years monthly data\footnote{All data is from Wharton Research Data Services (WRDS). https://wrds-web.wharton.upenn.edu/wrds/} from 08-31-1990 to 08-31-2000, and then test the learned allocation strategy from 09-29-2000 to 09-30-2010. The initial wealth is normalized as $1$ and the $10$ years target is $z=8$, corresponding to a $23\%$ annualized target return. In the daily rebalancing scenario, the EMV algorithm is trained on the $1$ year daily data from 01-09-2017 to 01-08-2018 and tested on the subsequent year, with a $40\%$ return set as the target for the $1$ year investment horizon.

For comparison studies, we also train and test other alternative methods for solving the portfolio allocation problem on the same data. Specifically, we consider the classical econometric methods including Black-Litterman (BL, \cite{black}), Fama-French (FF, \cite{fama}) and the Markowitz portfolio (Markowitz, \cite{M}). A recently developed distributionally robust MV strategy, the Robust Wasserstein Profile Inference (RWPI, \cite{robust_MV}), is also included. To compare EMV with deep RL method,  we adjust DDPG similarly as in \cite{hwang_2}, so that it can solve the classical MV problem (\ref{unconstrained_classical}). All experiments were performed on a MacBook Air laptop, with DDPG trained using Tensorflow.
\subsection{Test I: monthly rebalancing}
We first consider $d=20$. By randomly selecting $20$ stocks for each set/seed, we compose $100$ different seeds. The split of training and testing data for  EMV and DDPG is fixed as described above, but we consider two types of training. The first training method is batch (off-line) RL, where both algorithms are trained for multiple episodes using one seed, following by testing on the subsequent $10$ years data of that seed. The performance is then averaged over the $100$ seeds. Another method is to use all the $100$ seeds and select one seed randomly for each episode during training. Then both algorithms are tested on randomly selected $100$ seeds over the test period and the performance is averaged as well. The second method can be seen to artificially generate randomness for training and testing, and an algorithm that performs well using this method has universality and potential to generate to data of stocks in different sectors.

For competitive performance, we adopt a rolling-horizon based training and testing for all the other methods. Specifically, each time after the $1$ month ahead investment decision is made on the test set, we add the most recent price data point from the test set into the training set, and discard the most obsolete data point from the training set. 
\begin{figure}
\centering
  \begin{subfigure}[b]{0.4\textwidth}
    \includegraphics[width=\textwidth]{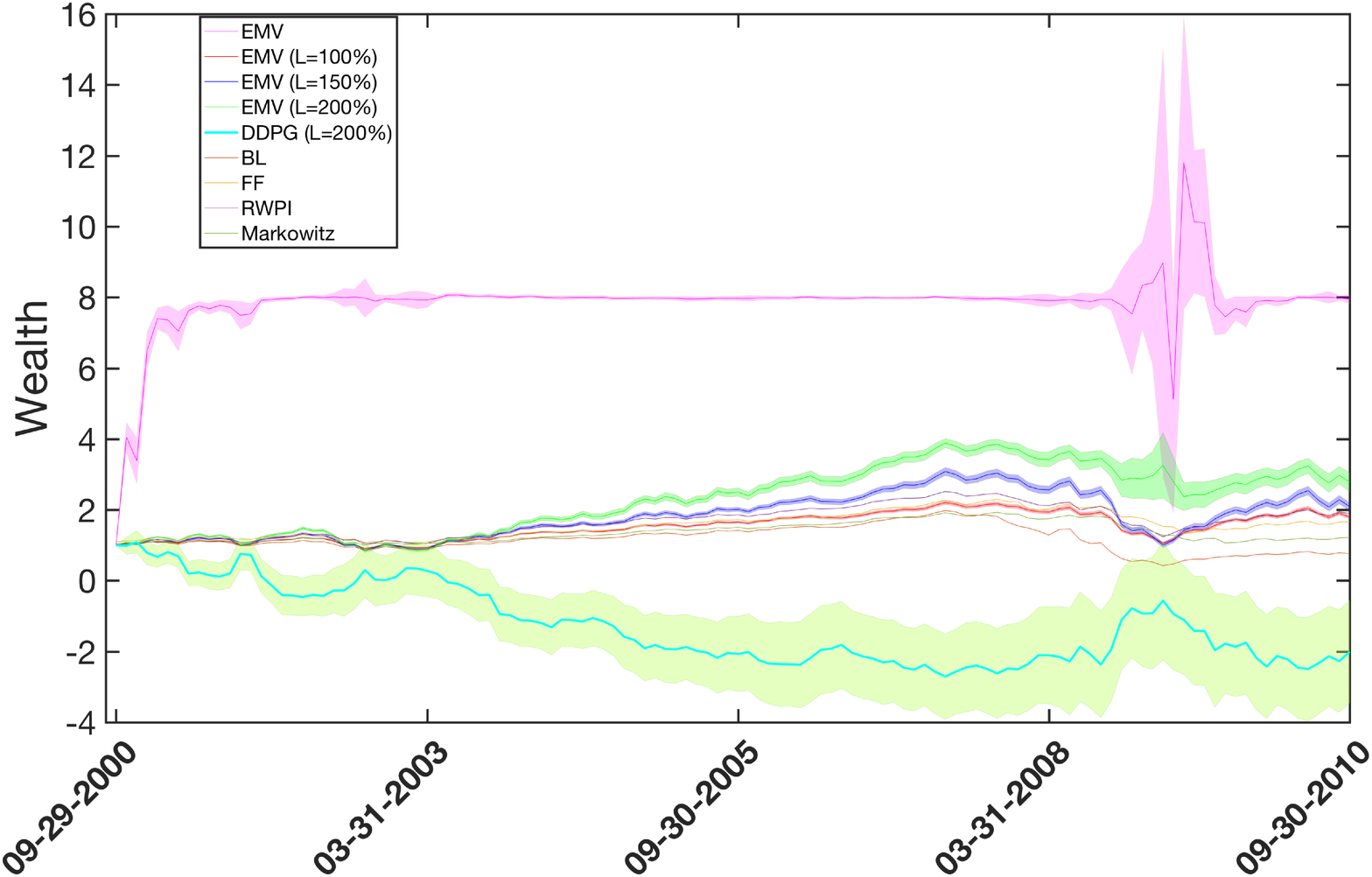}
    \caption{Monthly rebalancing ($d=20$)}
    \label{fig_monthly}
  \end{subfigure}
  \begin{subfigure}[b]{0.4\textwidth}
    \includegraphics[width=\textwidth]{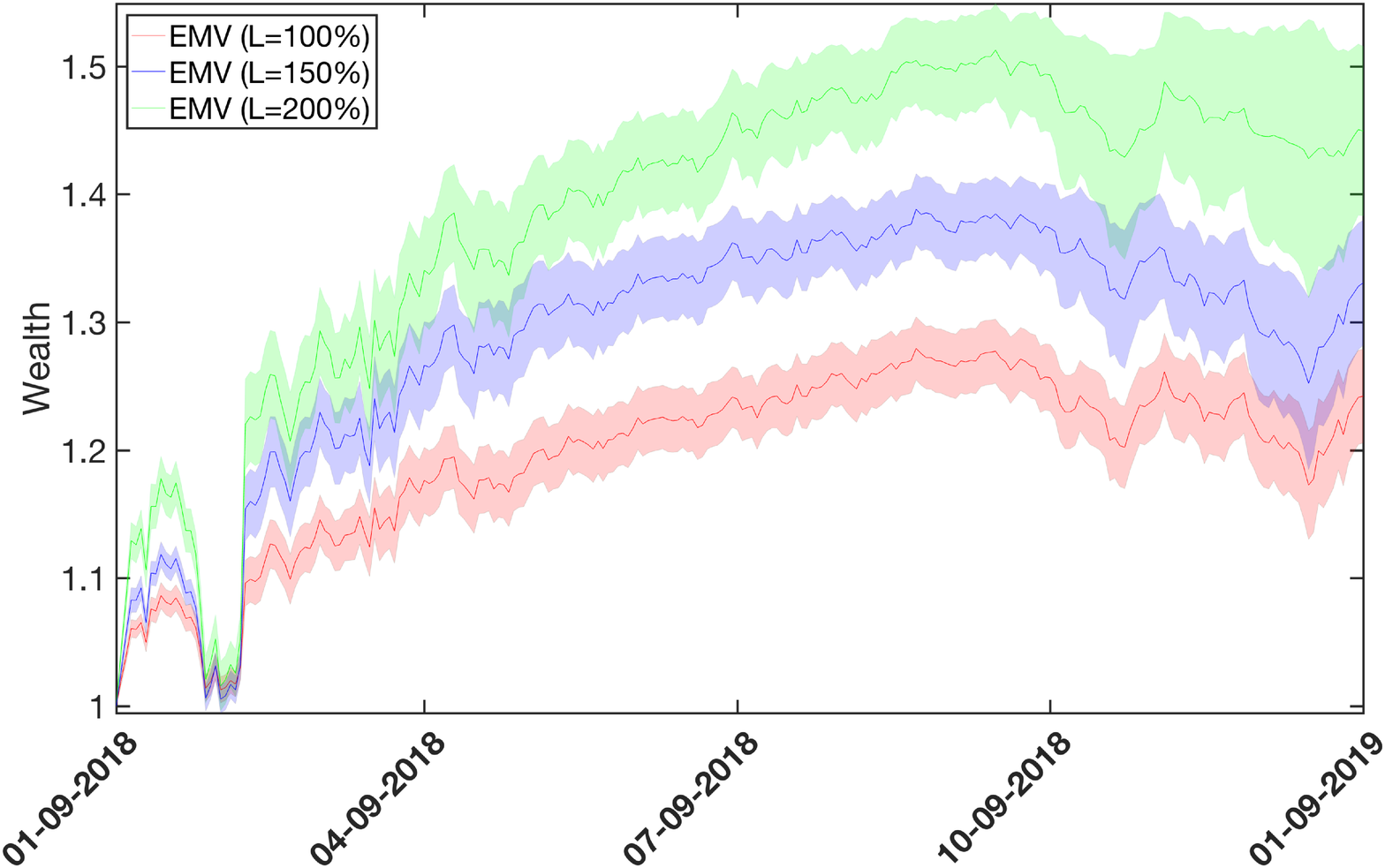}
    \caption{Daily rebalancing ($d=50$)}
    \label{fig_daily}
  \end{subfigure}
  \caption{Investment performance comparison for (a) $10$ years horizon with monthly rebalancing and (b) $1$ year horizon with daily rebalancing.}
\end{figure}
Figure \ref{fig_monthly} shows the performance of various investment strategies, including variants of the EMV algorithm with different {\it gross} leverage constraints on portfolios.\footnote{Leverage is a fundamental investment tool for most hedge funds; according to \cite{leverage}, the average gross leverage across the $208$ hedge funds studied therein is $213\%$.} Under reasonable leverage constraint, the EMV algorithm still outperforms most other methods (which have no constraints, except DDPG) by a large margin, although it was trained only using the previous $10$ years monthly data. 

The universal training and testing method was used for EMV and DDPG in Figure \ref{fig_monthly}. Results for the batch method can be found in Appendix \ref{batch}. A remarkable fact in both cases is that the original EMV algorithm, devised to solve the exploratory MV problem (\ref{value_function}) without constraint, achieves the target $z=8$ with minimal variance for most of the test period. We also report various investment outcomes in Table \ref{monthly_table} when scaling up $d$, the number of stocks in the portfolio.

\subsection{Test II: daily rebalancing} 
For daily trading with $d=50$, we report the performance of the EMV algorithm under different gross leverage constraints in Figure \ref{fig_daily}. The DDPG algorithm was not competitive in the daily trading setting (see Table \ref{daily_table}) and, hence, omitted. For different $d$, Table \ref{daily_table} summarizes the investment outcomes and the training time (per experiment). These results were obtained using the universal method for both training and testing.
\begin{table}
\scriptsize
\begin{subtable}[t]{\textwidth}\centering
\subcaption{\footnotesize Monthly rebalancing}
\begin{tabular}{@{}l cc cc cc} 
\toprule%
 & \multicolumn{2}{c}{{{\bfseries $d=20$}}}
 & \multicolumn{2}{c}{{{\bfseries  $d=60$}}}
 & \multicolumn{2}{c}{{{\bfseries $d=100$}}} \\

\cmidrule[0.4pt](r{0.155em}){1-1}%
\cmidrule[0.4pt](lr{0.125em}){2-3}%
\cmidrule[0.4pt](lr{0.125em}){4-5}%
\cmidrule[0.4pt](l{0.25em}){6-7}%

EMV (L=$200\%$) & $10.8\%$ & $0.31$ hrs & $11.2\%$ &
$4.34$ hrs & $6.3\%$ & $1.53$ hrs \textsuperscript{\ref{note1}}\\

& ($0.797$) &  & ($1.323$ ) &  & ($1.627$) & \\
[0.3cm]
DDPG (L=$200\%$) & $-300.1\%$  & $4.23$ hrs &
$476.3\%$  & $5.32$ hrs
& $-653.4\%$ & $6.68$ hrs  \\

\ \ (Unannualized) & ($-0.411$) &  & ($0.359$ ) &  & ($-0.432$) & \\

\bottomrule
\end{tabular}

\label{monthly_table}
\end{subtable}
\hspace{\fill}

\begin{subtable}[t]{\textwidth}\centering
\subcaption{\footnotesize Daily rebalancing}
\begin{tabular}{@{}l cc cc cc} 
\toprule%
 & \multicolumn{2}{c}{{{\bfseries $d=50$}}}
 & \multicolumn{2}{c}{{{\bfseries  $d=75$}}}
 & \multicolumn{2}{c}{{{\bfseries $d=100$}}} \\

\cmidrule[0.4pt](r{0.155em}){1-1}%
\cmidrule[0.4pt](lr{0.125em}){2-3}%
\cmidrule[0.4pt](lr{0.125em}){4-5}%
\cmidrule[0.4pt](l{0.25em}){6-7}%

EMV (L=$200\%$) & $44.9\%$ & $1.36$ hrs & $33.0\%$ &
$5.54$ hrs & $17.9\%$ & $1.63$ hrs\footnote{\label{note1}Only $1000$ episodes were trained, while the training episodes for other experiments were $20000$.}  \\

& ($1.347$) &  & ($1.370$) &  & ($1.124$) & \\
[0.3cm]
DDPG (L=$200\%$) & $-189.6\%$  & $6.20$ hrs &
$-27.9\%$  & $8.45$ hrs
& $-640.6\%$ & $14.42$ hrs  \\

& ($-0.096$) &  & ($-0.012$) &  & ($-0.219$) & \\

\bottomrule
\end{tabular}
\label{daily_table}
\end{subtable}
\caption{Annualized return (with Sharpe ratio) and training time corresponding to different $d$ for \\ (a) $10$ years horizon with monthly rebalancing and (b) $1$ year horizon with daily rebalancing.}
\end{table}

\section{Related Work}
The difficulty of seeking the global optimum for Markov Decision Process (MDP) problems under the MV criterion has been previously noted in \cite{MT}.
In fact, the variance of reward-to-go is nonlinear in expectation and, as a result of Bellman's inconsistency, most of the well-known RL algorithms cannot be applied directly.

Existing works on variance estimation and control generally divide into value based methods and policy based methods.  \cite{Sobel} obtained the Bellman's equation for the variance of reward-to-go under a {\it fixed}, given policy. \cite{SKK} further derived the TD(0) learning rule to estimate the variance, followed by \cite{SK} which applied this value based method  to an MV  portfolio selection problem. It is worth noting that due to the definition of the value function (i.e., the variance penalized expected reward-to-go) in \cite{SK},  Bellman's optimality principle does not hold. As a result, it is not guaranteed that a greedy policy based on the latest updated value function will eventually lead to the true global optimal policy. The second approach, the policy based RL,  was proposed  in \cite{TDM}. They also extended the work to linear function approximators and devised actor-critic algorithms for  MV optimization problems for which convergence to the local optimum is guaranteed with probability one (\cite{TM}). Related works following this line of research include \cite{MV_2}, \cite{MV_1}, among others. Despite the various methods mentioned above, it remains  an open and interesting question in RL to search for the {\it global} optimum under the MV criterion.

In this paper, rather than relying on the typical framework of discrete-time MDP and discretizing time and state/action spaces accordingly, we designed the EMV algorithm to learn the {\it global} optimal solution of the continuous-time exploratory MV problem (\ref{value_function}) {\it directly}. As pointed out in \cite{Doya}, it is typically challenging to find the right granularity to discretize the state and action spaces, and
naive discretization  may lead to poor performance. On the other hand, grid-based discretization methods for solving the HJB equation cannot easily extend to high dimensions in practice due to the curse of dimensionality, although theoretical convergence results have been established (see \cite{Munos_1}, \cite{Munos_2}). Our EMV algorithm, however, is computationally feasible and implementable in high dimensions, as demonstrated by the experiements, due to the explicit representations of the value functions and the portfolio strategies, thereby devoid of  the curse of dimensionality. Note that our algorithm does not use (deep) neural networks, which have been applied extensively in literature for (high-dimensional) continuous RL problems (e.g., \cite{Lillicrap}, \cite{DQN}) but known for unstable performance, sample inefficiency as well as extensive hyperparameter tuning (\cite{DQN}, \cite{Deep_RL_1}, \cite{Deep_RL_2}), in addition to their low interpretability.\footnote{Interpretability is one of the most important and pressing issues in the general artificial intelligence applications in financial industry due to, among others,
the regulatory requirement.}

\section{Conclusions}
We studied continuous-time mean-variance (MV) portfolio allocation problem in high dimensions using RL methods. Under the exploratory control framework for general continuous-time optimization problems, we formulated the exploratory MV problem in high dimensions and proved the optimality of Gaussian policy in achieving the best tradeoff between exploration and exploitation. Our EMV algorithm, designed by combining quantitative finance analysis and RL techniques to solve the exploratory MV problem, is interpretable, scalable and data efficient, thanks to a provable policy improvement theorem and efficient functional approximations based on the theoretical optimal solutions. It consistently outperforms both classical model-based econometric methods and model-free deep RL method, across different training and testing scenarios. Interesting future research includes testing the EMV algorithm for shorter trading horizons with tick data (e.g. high frequency trading), or for trading other financial instruments such as mean-variance option hedging.

\subsubsection*{Acknowledgments}
The author would like to thank Prof. Xun Yu Zhou for generous support and continuing encouragement on this work. The author also wants to thank Lin (Charles) Chen for providing the results on BL, FF, Markowitz and RWPI methods.

\bibliography{bibtex_2}

\newpage
\appendix

\section{Controlled Wealth Dynamics}\label{wealth_appendix}
Let $W_t=(W^1_t,\dots, W^d_t)$, $0\leq t\leq T$ be a standard $d$-dimensional Brownian motion defined on a filtered probability space $(\Omega, \mathcal{F},\{\mathcal{F}_{t}\}_{0\leq t\leq T},\mathbb{P})$ that satisfies the usual conditions. The price process of the $i$-th risky asset is a geometric Brownian motion governed by
\begin{equation}\label{price}
dS^{i}_t=S^{i}_t\left(\mu^i\, dt+\sigma^{i} \cdot dW_t\right), \quad 0\leq t\leq T, \quad  i=1,\dots,d,
\end{equation}
with $S^i_0=s^i_0>0$ being the initial price at $t=0$, and $\mu^i\in \mathbb{R}$, $\sigma^i=(\sigma^{1i},\dots, \sigma^{di})\in \mathbb{R}^d$ being the mean return and volatility coefficients of the $i$-th risky asset, respectively. We denote for brevity the mean return vector by $\mu\in \mathbb{R}^d$, and the volatility matrix by $\sigma\in \mathbb{R}^{d\times d}$, whose $i$-th column represents the volatility $\sigma^i$ of the $i$-th risky asset. The riskless asset has a constant interest rate $r>0$. We assume that $\sigma$ is non-degenerate and hence there exists a $d$-dimensional vector $\rho$ that satisfies $\sigma'\rho=\mu-r\mathbf{1}$, where $\mathbf{1}$ is the $d$-dimensional vector with all components being $1$. The vector $\rho$ is known as the market price of risk. It is worth noting that the above assumptions are only made for the convenience of deriving theoretical results in the paper; in practice, all the model parameters are unknown and time-varying, and it is the goal of RL algorithms to directly output trading strategies without relying on estimation of any underlying parameters.

Denote by $u^0_t$ and $u_t=(u_t^1,\dots,u_t^d)$  the discounted dollar value put in the savings account and the $d$ risky assets, respectively, at time $t$. It then follows that the discounted wealth process is $x^u_t=\sum_{i=0}^d u^i_t$, $0\leq t\leq T$. The self-financing condition further implies that, using (\ref{price}), we have
$$dx^u_t =  ru^0_t dt+\sum_{i=1}^d\frac{u^i_t}{S^i_t}dS^i_t-rx^u_t dt=-r(x^u_t-u^0_t)dt+\sum_{i=1}^d u^i_t\left(\mu^i dt+\sigma^i\cdot dW_t\right)$$
$$=\sum_{i=1}^d u_t^i\left((\mu^i-r)dt+\sigma^i\cdot dW_t\right)=\sigma u_t\cdot (\rho dt+dW_t).$$
\section{Value Functions and Admissible Control Distributions}\label{value_function_def}
In order to rigorously solve (\ref{value_function}) by dynamic programming, we need to define the value functions. For each $(s,y)\in[0,T)\times \mathbb{R}$, consider the
state equation (\ref{state_process0}) on $[s,T]$ with $X^{\pi}_s=y$.
Define the set of admissible controls, $\mathcal{A}(s,y)$, as follows. Let $\mathcal{B}(\mathbb{R}^d)$ be the Borel algebra on $\mathbb{R}^d$. A (distributional) control (or portfolio/strategy) process $\pi=\{\pi_t,s\leq t\leq T\}$ belongs to $\mathcal{A}(s,y)$, if

\smallskip

(i)		\ for each $s\leq t\leq T$, $\pi _{t}\in \mathcal{P}(\mathbb{R}^d)$ a.s.;

(ii)		for each $A\in \mathcal{B}(\mathbb{R}^d)$, $\{\int_A\pi _{t}(u)du,s\leq t\leq T\} $
is $\mathcal{F}_{t}$-progressively measurable;

(iii)	$\mathbb{E}\left[
\int_{s}^{T}\int_{\mathbb{R}^d}\big|\sigma u\big|^2\, \pi_t(u)du dt\right] <\infty$;

(iv)	$\mathbb{E}\left[\big|(X_T^{\pi}-w)^2+\lambda \int_s^T\int_{\mathbb{R}^d}\pi_t(u)\ln \pi_t(u)dudt\big|\; \Big | X_s^{\pi}=y\right]<\infty$.

\medskip

Clearly, it follows from condition (iii) that
the stochastic differential equation (SDE) (\ref{state_process0}) has a unique strong solution for $s\leq t\leq T$ that satisfies $X^{\pi}_s=y$.

Controls in $\mathcal{A}(s,y)$ are measure-valued (or, precisely, density-function-valued) stochastic {\it processes}, which are also called {\it open-loop} controls in the control terminology.
As in the classical control theory, it is important to distinguish between open-loop controls and {\it feedback} (or {\it closed-loop}) controls (or {\it policies} as in the RL literature, or {\it laws} as in the control literature). Specifically, a {\it deterministic} mapping $\boldsymbol{\pi}(\cdot;\cdot,\cdot)$ is called an (admissible)  feedback control  if i) $\boldsymbol{\pi}(\cdot;t,x)$ is a density function for each $(t,x)\in[0,T]\times \mathbb{R}$; ii) for each $(s,y)\in[0,T)\times \mathbb{R}$, the following SDE (which is the system dynamics after the feedback policy $\boldsymbol{\pi}(\cdot;\cdot,\cdot)$ is applied)
\begin{equation}\label{new_dynamics_feedback}
dX^{\boldsymbol{\pi}}_t=\left(\int_{\mathbb{R}^d} \rho' \sigma u\boldsymbol{\pi}(u;t,X^{\boldsymbol{\pi}}_t))du\right)\, dt+\left({\int_{\mathbb{R}^d} u'\sigma'\sigma u\boldsymbol{\pi}(u;t,X^{\boldsymbol{\pi}}_t))du}\right)^{\frac{1}{2}}\, dB_t,\;  X^{\boldsymbol{\pi}}_{s}=y,
\end{equation}
has a unique strong solution $\{X^{\boldsymbol{\pi}}_t,t\in[s,T]\}$,  and  the open-loop control
$\pi=\{\pi
_{t},$ $t\in[s,T]\}\in \mathcal{A}(s,y)$ where $\pi_{t}:=\boldsymbol{\pi}(\cdot;t,X_t^{\boldsymbol{\pi}})$. In this case, the open-loop control $\pi$ is said to be
{\it generated} from the feedback policy $\boldsymbol{\pi}(\cdot;\cdot,\cdot)$ {\it with respect to} the initial time and state,  $(s,y)$. It is useful to note that an open-loop control and its admissibility depend on the initial $(s,y)$, whereas a
feedback policy can generate open-loop controls for {\it any} $(s,y)\in[0,T)\times \mathbb{R}$, and hence is in itself independent of $(s,y)$. Note that throughout this paper, we have used boldfaced $\boldsymbol{\pi}$ to denote feedback controls, and the normal style $\pi$ to denote open-loop controls.

Now, for a fixed $w\in \mathbb{R}$, define
\begin{equation}\label{value_function_general}
V(s,y;w):=\inf_{\pi\in \mathcal{A}(s,y)}\mathbb{E}\left[(X_T^{\pi}-w)^2+\lambda \int_0^T\int_{\mathbb{R}^d}\pi_t(u)\ln \pi_t(u)dudt\Big | X_s^{\pi}=y\right]-(w-z)^2,
\end{equation}
for $(s,y)\in[0,T)\times \mathbb{R}$.
The function $V(\cdot,\cdot;w)$ is called the {\it optimal  value function} of the problem.

Moreover, we define  the {\it value function}  under any given {\it feedback} control  $\boldsymbol{\pi}$:

\begin{equation}\label{general_value}
V^{\boldsymbol{\pi}}(s,y ;w)=\mathbb{E}\left[(X_T^{\boldsymbol{\pi}}-w)^2+\lambda \int_s^T\int_{\mathbb{R}^d}\pi_t(u)\ln \pi_t(u)dudt\Big | X_s^{\boldsymbol{\pi}}=y\right]-(w-z)^2,
\end{equation}
for $(s,y)\in[0,T)\times \mathbb{R}$, where $\pi=\{\pi
_{t},$ $t\in[s,T]\}$ is  the open-loop control generated from $\boldsymbol{\pi}$  with respect to  $(s,y)$ and $\{X^{\boldsymbol{\pi}}_t,t\in[s,T]\}$ is the corresponding wealth process.

Note that in the control literature, $V$ given by (\ref{value_function_general}) is called the value function.
However, in the RL literature the term ``value function" is also used for the objective value under a particular control (i.e. $V^{\boldsymbol{\pi}}$ in (\ref{general_value})). So to avoid ambiguity we have called $V$ the {\it optimal} value function in this paper.

\section{Proofs}
\subsection{Proof of Theorem \ref{verification}}\label{proof1}

The main proof of Theroem \ref{verification} would be the verification arguments that aim to show the optimal value function of problem (\ref{value_function}) is given by (\ref{value_verified}) and that the candidate optimal policy (\ref{Gaussian_verified}) is indeed admissible, based on the definitions in Appendix \ref{value_function_def}. Since the current exploratory MV problem is a special case of the exploratory linear-quadratic problem extensively studied in \cite{Hwang}, a detailed proof would follow the same lines of that of Theorem $4$ therein, and is left for interested readers.

\smallskip

\begin{proof}
We now determine the Lagrange multiplier $w$ through the constraint $\mathbb{E}[X^*_T]=z$. It follows  from (\ref{wealth_SDE_1}),
along with the standard estimate that $\mathbb{E}\left[\max_{t\in [0,T]}(X^*_t)^2\right]<\infty$ and Fubini's Theorem, that
$$\mathbb{E}[X^*_t]=x_0+\mathbb{E}\left[\int_0^t -\rho'\rho(X^*_s-w)\, ds\right]=x_0+\int_0^t -\rho'\rho\left(\mathbb{E}[X^*_s]-w\right)\, ds.$$
Hence,  $\mathbb{E}[X^*_t]=(x_0-w)e^{-\rho'\rho t}+w$. The constraint $\mathbb{E}[X^*_T]=z$ now becomes $(x_0-w)e^{-\rho'\rho T}+w=z$, which gives  $w=\frac{ze^{\rho'\rho T}-x_0}{e^{\rho'\rho T}-1}$.
\end{proof}

\subsection{Proof of Theorem \ref{convergence_to_Dirac}}\label{proof2}
To prove the solution of the exploratory MV problem converges to that of the classical MV problem, as $\lambda\rightarrow 0$, we first recall the solution of the classical MV problem. 

In order to apply dynamic programming for (\ref{unconstrained_classical}), we again consider the set of admissible controls, $\mathcal{A}^{\text{cl}}(s,y)$,
for $(s,y)\in [0,T)\times \mathbb{R}$,
\begin{center}
$\mathcal{A}^{\text{cl}}(s,y):=\Big\{u=\{u_t, t\in [s,T]\}$: $u$ is $\mathcal{F}_t$-progressively measurable and $\mathbb{E}\left[\int_s^T\big|\sigma u_t\big|^2\, dt\right]<\infty\Big\}.$
\end{center}
The (optimal) value function is defined by
\begin{equation}\label{classical_value_function}
V^{\text{cl}}(s,y;w):=\inf_{u\in \mathcal{A}^{\text{cl}}(s,y)}\mathbb{E}\big [(x^u_T-w)^2\, \big | \, x^u_s=y\big ]-(w-z)^2,
\end{equation}
for $(s,y)\in [0,T)\times \mathbb{R}$, where $w\in \mathbb{R}$ is fixed. Once this problem is solved, $w$ can be determined by the constraint $\mathbb{E}[x^*_T]=z$, with $\{x^*_t, t\in [0,T]\}$ being the optimal wealth process under the optimal portfolio $u^*$.

The HJB equation is
\begin{equation}\label{HJB_classical}
\omega_t(t,x;w)+\min_{u\in \mathbb{R}^d}\left(\frac{1}{2}u'\sigma'\sigma u\,\omega_{xx}(t,x;w)+\rho' \sigma u \, \omega_x(t,x;w)\right)=0,\;\;(t,x)\in [0,T)\times \mathbb{R},
\end{equation}
with the terminal condition  $\omega(T,x;w)=(x-w)^2-(w-z)^2$.

Standard verification arguments  deduce the optimal value function to be
\begin{equation}\label{classical_value_verified}
V^{\text{cl}}(t,x;w)=(x-w)^2e^{-\rho'\rho (T-t)}-(w-z)^2,
\end{equation}
the optimal feedback control policy to be
\begin{equation}\label{classical_optimal_control}
\boldsymbol{u}^{\ast }(t,x;w)=-{\sigma}^{-1}\rho(x-w),
\end{equation}
and the corresponding optimal wealth process to be the unique strong solution to the SDE
\begin{equation}\label{classical_optimal_wealth}
dx^*_t=-\rho'\rho(x^*_t-w)\, dt-\rho (x^*_t-w)\cdot dW_t,\quad x^*_0=x_0.
\end{equation}
Comparing the optimal wealth dynamics, (\ref{wealth_SDE_1}) and (\ref{classical_optimal_wealth}), of the exploratory and classical problems,
we note that they have the same {\it drift} coefficient (but different {\it diffusion} coefficients). As a result, the two problems have the same
mean of optimal terminal wealth and hence the same value of
the Lagrange multiplier $w=\frac{ze^{\rho'\rho T}-x_0}{e^{\rho'\rho T}-1}$ determined by the constraint $\mathbb{E}[x^*_T]=z$.

\smallskip

\begin{proof}
The weak convergence of the feedback controls follows from the explicit forms of $\boldsymbol{\pi}^{\ast}$  in (\ref{Gaussian_verified}) and $\boldsymbol{u}^{\ast}$ in (\ref{classical_optimal_control}). The pointwise convergence of the value functions follows easily from the forms of $V$ in (\ref{value_verified}) and $%
V^{\text{cl}}$ in (\ref{classical_value_verified}), together with the fact that
\[ \lim_{\lambda \rightarrow 0}\frac{\lambda}{2}\ln\frac{|\sigma'\sigma|}{\pi \lambda} =0.
\]
\end{proof}

\subsection{Proof of Theorem \ref{PIT}}\label{proof3}
\begin{proof}
Fix $(t,x)\in [0,T]\times \mathbb{R}$. Since, by assumption, the feedback policy $\tilde{{\boldsymbol{\pi}}}$ is admissible,
the open-loop control strategy, $\tilde{\pi}=\{\tilde{\pi}_v, v\in [t,T]\}$, generated from $\tilde{{\boldsymbol{\pi}}}$ with respect to the initial condition ${X}^{\tilde{{\boldsymbol{\pi}}}}_t=x$ is  admissible. Let $\{{X}^{\tilde{{\boldsymbol{\pi}}}}_s,s\in[t,T]\}$ be the corresponding wealth process under $\tilde{\pi}$. Applying It\^{o}'s formula, we have
$$V^{\boldsymbol{\pi}}(s,\tilde{X}_s)=V^{\boldsymbol{\pi}}(t,x)+\int_t^sV^{\boldsymbol{\pi}}_t(v,{X}^{\tilde{{\boldsymbol{\pi}}}}_v)dv+\int_t^s\int_{\mathbb{R}^d}\Big(\frac{1}{2}u'\sigma'\sigma u V^{\boldsymbol{\pi}}_{xx}(v,{X}^{\tilde{{\boldsymbol{\pi}}}}_v)$$
\begin{equation}\label{Ito}
+\rho' \sigma u V^{\boldsymbol{\pi}}_x(v,{X}^{\tilde{{\boldsymbol{\pi}}}}_v)\Big)\tilde{\pi}_v(u) \, dudv+\int_t^s  \left(\int_{\mathbb{R}^d}u'\sigma'\sigma u\tilde{\pi}_v(u)du\right)^{\frac{1}{2}}V^{\boldsymbol{\pi}}(v,{X}^{\tilde{{\boldsymbol{\pi}}}}_v)\, dB_v,\;s\in [t,T].
\end{equation}
Define the stopping times $\tau_n:=\inf\{s\geq t: \int_t^s \int_{\mathbb{R}^d}u'\sigma'\sigma u\,\tilde{\pi}_v(u)du\left(V^{\boldsymbol{\pi}}(v,{X}^{\tilde{{\boldsymbol{\pi}}}}_v)\right)^2dv\geq n\}$, for $n\geq 1$. Then, from (\ref{Ito}), we obtain
$$V^{\boldsymbol{\pi}}(t,x)=\mathbb{E}\Big[V^{\boldsymbol{\pi}}(s\wedge\tau_n,{X}^{\tilde{{\boldsymbol{\pi}}}}_{s\wedge \tau_n})-\int_t^{s\wedge\tau_n}V^{\boldsymbol{\pi}}_t(v,{X}^{\tilde{{\boldsymbol{\pi}}}}_v)dv$$
\begin{equation}\label{stopped}
-\int_t^{s\wedge\tau_n}\int_{\mathbb{R}^d}\Big(\frac{1}{2}u'\sigma'\sigma u \,V^{\boldsymbol{\pi}}_{xx}(v,{X}^{\tilde{{\boldsymbol{\pi}}}}_v)+\rho' \sigma u V^{\boldsymbol{\pi}}_x(v,{X}^{\tilde{{\boldsymbol{\pi}}}}_v)\Big)\tilde{\pi}_v (u)\, dudv\, \Big| {X}^{\tilde{{\boldsymbol{\pi}}}}_t=x\Big].
\end{equation}
On the other hand, by standard arguments and the assumption that $V^{\boldsymbol{\pi}}$ is smooth, we have
$$V_t^{\boldsymbol{\pi}}(t,x)+\int_{\mathbb{R}^d}\left(\frac{1}{2}u'\sigma'\sigma u\, V^{\boldsymbol{\pi}}_{xx}(t,x)+\rho'\sigma u V^{\boldsymbol{\pi}}_x(t,x)+\lambda\ln\boldsymbol{\pi}(u;t,x)\right)\boldsymbol{\pi}(u;t,x)du=0,$$
for any $(t,x)\in [0,T)\times \mathbb{R}$. It follows that
\begin{equation}\label{minimization}
V_t^{\boldsymbol{\pi}}(t,x)+\min_{\hat{\pi}\in\mathcal{P}(\mathbb{R}^d)}\int_{\mathbb{R}^d}\left(\frac{1}{2}u'\sigma'\sigma u\, V^{\boldsymbol{\pi}}_{xx}(t,x)+\rho'\sigma u V^{\boldsymbol{\pi}}_x(t,x)+\lambda\ln\hat{\pi}(u)\right)\hat{\pi}(u)du\leq 0.
\end{equation}
Notice that the minimizer of the Hamiltonian in (\ref{minimization}) is given by the feedback policy $\tilde{\boldsymbol{\pi}}$ in (\ref{new_policy}). It then follows that equation (\ref{stopped}) implies
$$V^{\boldsymbol{\pi}}(t,x)\geq\mathbb{E}\Big[V^{\boldsymbol{\pi}}(s\wedge\tau_n,{X}^{\tilde{{\boldsymbol{\pi}}}}_{s\wedge \tau_n})+\lambda \int_t^{s\wedge\tau_n}\int_{\mathbb{R}^d}\tilde{\pi}_v(u)\ln\tilde{\pi}_v(u)\, dudv\Big | {X}^{\tilde{{\boldsymbol{\pi}}}}_t=x\Big],$$
for $(t,x)\in [0,T]\times \mathbb{R}$ and $s\in [t,T]$. Now taking $s=T$, and using that $V^{\boldsymbol{\pi}}(T,x)=V^{\tilde{{\boldsymbol{\pi}}}}(T,x)=(x-w)^2-(w-z)^2$ together with the assumption that $\tilde{\pi}$ is admissible, we obtain, by sending  $n\rightarrow \infty$ and applying the dominated convergence theorem, that
$$V^{\boldsymbol{\pi}}(t,x)\geq\mathbb{E}\Big[V^{\tilde{{\boldsymbol{\pi}}}}(T,{X}^{\tilde{{\boldsymbol{\pi}}}}_{T})+\lambda \int_t^{T}\int_{\mathbb{R}^d}\tilde{\pi}_v(u)\ln\tilde{\pi}_v(u)\, dudv\Big | {X}^{\tilde{{\boldsymbol{\pi}}}}_t=x\Big]=V^{\tilde{{\boldsymbol{\pi}}}}(t,x),$$
for any $(t,x)\in [0,T]\times \mathbb{R}$.
\end{proof}

\subsection{Proof of Theorem \ref{convergence_learning}}\label{proof4}
\begin{proof}
It can be easily verified that the feedback policy $\boldsymbol{\pi}_0(u;t,x,w)=\mathcal{N}(u| \alpha(x-w), \Sigma e^{\beta(T-t)})$ generates an open-loop policy $\pi_0$ that is admissible with respect to the initial $(t,x)$. Moreover, it follows from the Feynman-Kac formula that the corresponding value function $V^{\boldsymbol{\pi}_0}$ satisfies the PDE
$$V^{\boldsymbol{\pi}_0}_t(t,x;w)+\int_{\mathbb{R}^d}\Big(\frac{1}{2}u'\sigma'\sigma u\, V^{\boldsymbol{\pi}_0}_{xx}(t,x;w)+\rho' \sigma  u V^{\boldsymbol{\pi}_0}_x(t,x;w)$$
\begin{equation}
+\lambda \ln {\pi}_0(u;t,x,w)\Big){\pi}_0(u;t,x,w)du=0,
\end{equation}
with terminal condition $V^{\boldsymbol{\pi}_0}(T,x;w)=(x-w)^2-(w-z)^2$.
Simplifying this equation we obtain
$$V_t^{\boldsymbol{\pi}_0}(t,x;w)+\frac{1}{2}V_{xx}^{\boldsymbol{\pi}_0}(t,x;w)\boldsymbol{\text{Tr}}\left(\sigma \alpha\alpha'\sigma (x-w)^2+\sigma \Sigma \sigma' e^{\beta (T-t)}\right)$$
\begin{equation}\label{first_value}
+V_{x}^{\boldsymbol{\pi}_0}(t,x;w)\rho'\sigma \alpha (x-w)-\frac{\lambda}{2}\big(d\ln(2\pi e)+\ln|\Sigma|+d\beta (T-t)\big)=0,
\end{equation}
where $\boldsymbol{\text{Tr}}(\cdot)$ denotes the trace of a square matrix. A classical solution to equation (\ref{first_value}) is given by
$$V^{\boldsymbol{\pi}_0}=(x-w)^2 e^{(2\rho'\sigma \alpha+\boldsymbol{\text{Tr}}(\sigma\alpha\alpha'\sigma'))(T-t)}+\frac{\boldsymbol{\text{Tr}}(\sigma \Sigma \sigma')e^{(\beta+2\rho'\sigma \alpha+\boldsymbol{\text{Tr}}(\sigma\alpha\alpha'\sigma'))(T-t)}}{\beta+2\rho'\sigma \alpha+\boldsymbol{\text{Tr}}(\sigma\alpha\alpha'\sigma')}$$
$$-\frac{\lambda d}{4}\beta t^2+\frac{\lambda d}{2}\left(\ln\left(2\pi e |\Sigma|^{\frac{1}{d}}\right)+\beta T\right)t-(w-z)^2-\frac{\boldsymbol{\text{Tr}}(\sigma \Sigma \sigma')}{\beta+2\rho'\sigma \alpha+\boldsymbol{\text{Tr}}(\sigma\alpha\alpha'\sigma')},$$
if $\beta+2\rho'\sigma \alpha+\boldsymbol{\text{Tr}}(\sigma\alpha\alpha'\sigma')\neq 0$ and, by
$$V^{\boldsymbol{\pi}_0}=(x-w)^2 e^{(2\rho'\sigma \alpha+\boldsymbol{\text{Tr}}(\sigma\alpha\alpha'\sigma'))(T-t)}-\frac{\lambda d}{4}\beta t^2+\left(\frac{\lambda d}{2}\left(\ln\left(2\pi e |\Sigma|^{\frac{1}{d}}\right)+\beta T\right)-\boldsymbol{\text{Tr}}(\sigma \Sigma \sigma')\right)t$$
$$-(w-z)^2,$$
if $\beta+2\rho'\sigma \alpha+\boldsymbol{\text{Tr}}(\sigma\alpha\alpha'\sigma')= 0$.
In either case, it is easy to check  that $V^{\boldsymbol{\pi}_0}$ satisfies the conditions in Theorem \ref{PIT} and, hence, the theorem applies. The improved policy is given by (\ref{new_policy}), which, in the current case, becomes
$$\boldsymbol{\pi}_1(u;t,x,w)=\mathcal{N}\left( \, u\, \Big| -\sigma^{-1}\rho (x-w)\ , \  \frac{\lambda (\sigma'\sigma)^{-1}}{2e^{(2\rho'\sigma \alpha+\boldsymbol{\text{Tr}}(\sigma\alpha\alpha'\sigma'))(T-t)}} \right).$$
Again, we can calculate the corresponding value function as $V^{\boldsymbol{\pi}_1}(t,x;w)=(x-w)^2e^{-\rho'\rho (T-t)}+F_1(t)$, where $F_1$ is a function of $t$ only. Theorem \ref{PIT} is applicable again, which yields the improved policy $\boldsymbol{\pi}_2$ as exactly the optimal Gaussian policy $\boldsymbol{\pi^*}$ given in (\ref{Gaussian_verified}), together with the optimal value function $V$ in (\ref{value_verified}). The desired convergence therefore follows, as for $n\geq 2$, both the policy and the value function will no longer strictly improve  under the policy improvement scheme (\ref{new_policy}).
\end{proof}

\section{Empirical Results: the Batch Method}\label{batch}
In Section 5.2, we provided the experiment results for monthly trading under universal training and testing. Another way to train and test the EMV and DDPG algorithms is based on batch (off-line) RL, as described in the main text. The batch method applies to the training and testing data from the {\it same} set/seed of $d=20$ stocks for each experiment, and the investment performance of the EMV algorithm over the $100$ seeds are reported in Figure \ref{fig_1}. Due to the extensive training time (see Table \ref{monthly_table}), we only train and test DDPG under the batch method for $8$ seeds. 

The batch method demonstrates qualitatively similar behavior as the universal training and testing method (see Figure \ref{fig_monthly}), when compared to the econometric methods and the deep RL method. A more detailed comparison between the two methods is shown in Figure \ref{fig_2}. It is interesting to notice that, while the two methods perform equally well for most of the testing period over $2000-2010$, the universal method is less affected by the $2008$ financial crisis with less variability and higher returns. The batch method, without taking into account the data of other stocks for each portfolio/seed during training and testing, is more susceptible to stock market plunge. Nonetheless, both methods are data efficient, especially in view that, for example, the training set for the batch method contains the same number of data points as the testing set decision making points ($120\times 20$).

\begin{figure}
\centering
  \begin{subfigure}[b]{0.45\textwidth}
    \includegraphics[width=\textwidth]{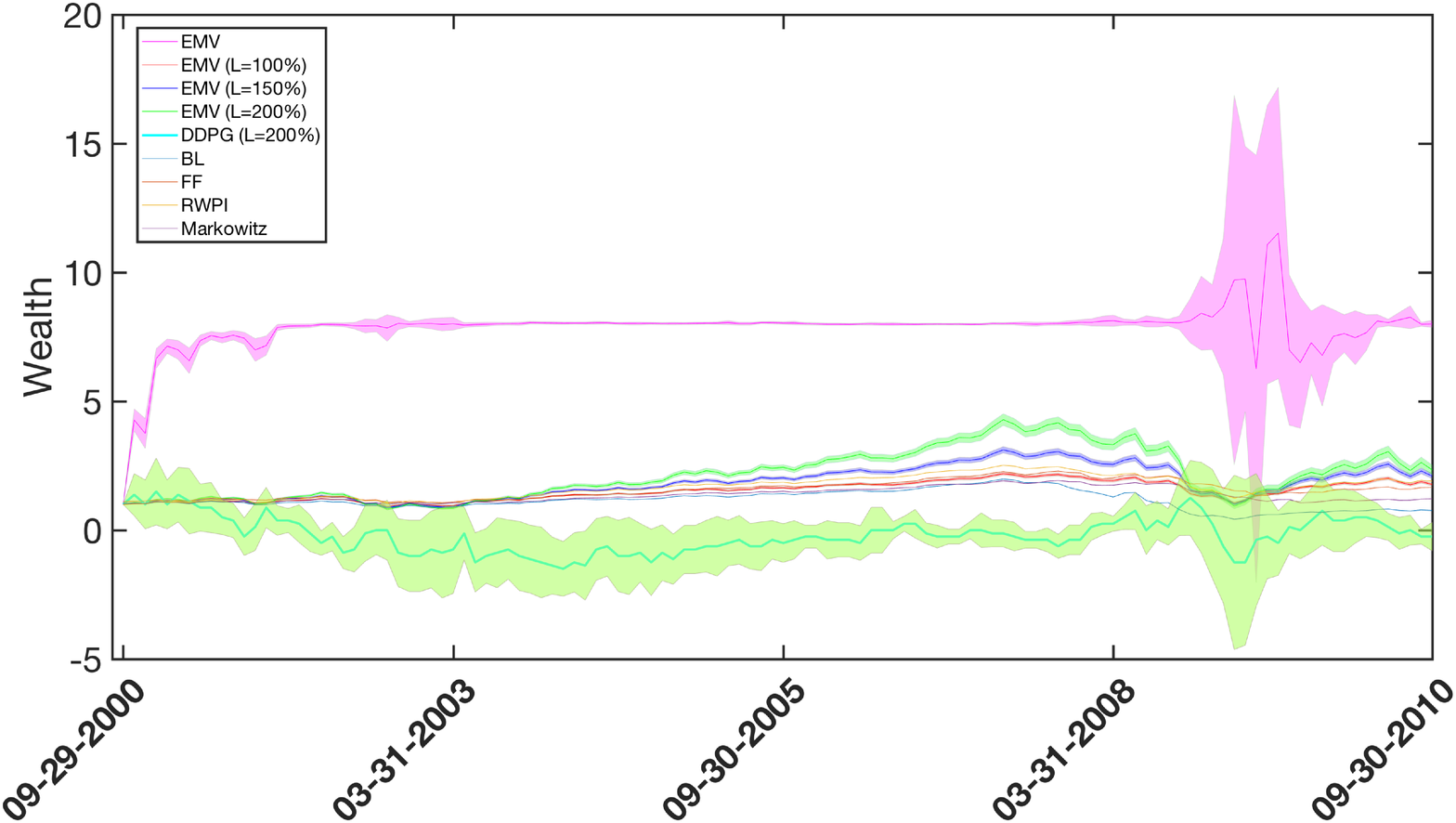}
    \caption{Batch method}
    \label{fig_1}
  \end{subfigure}
  \begin{subfigure}[b]{0.45\textwidth}
    \includegraphics[width=\textwidth]{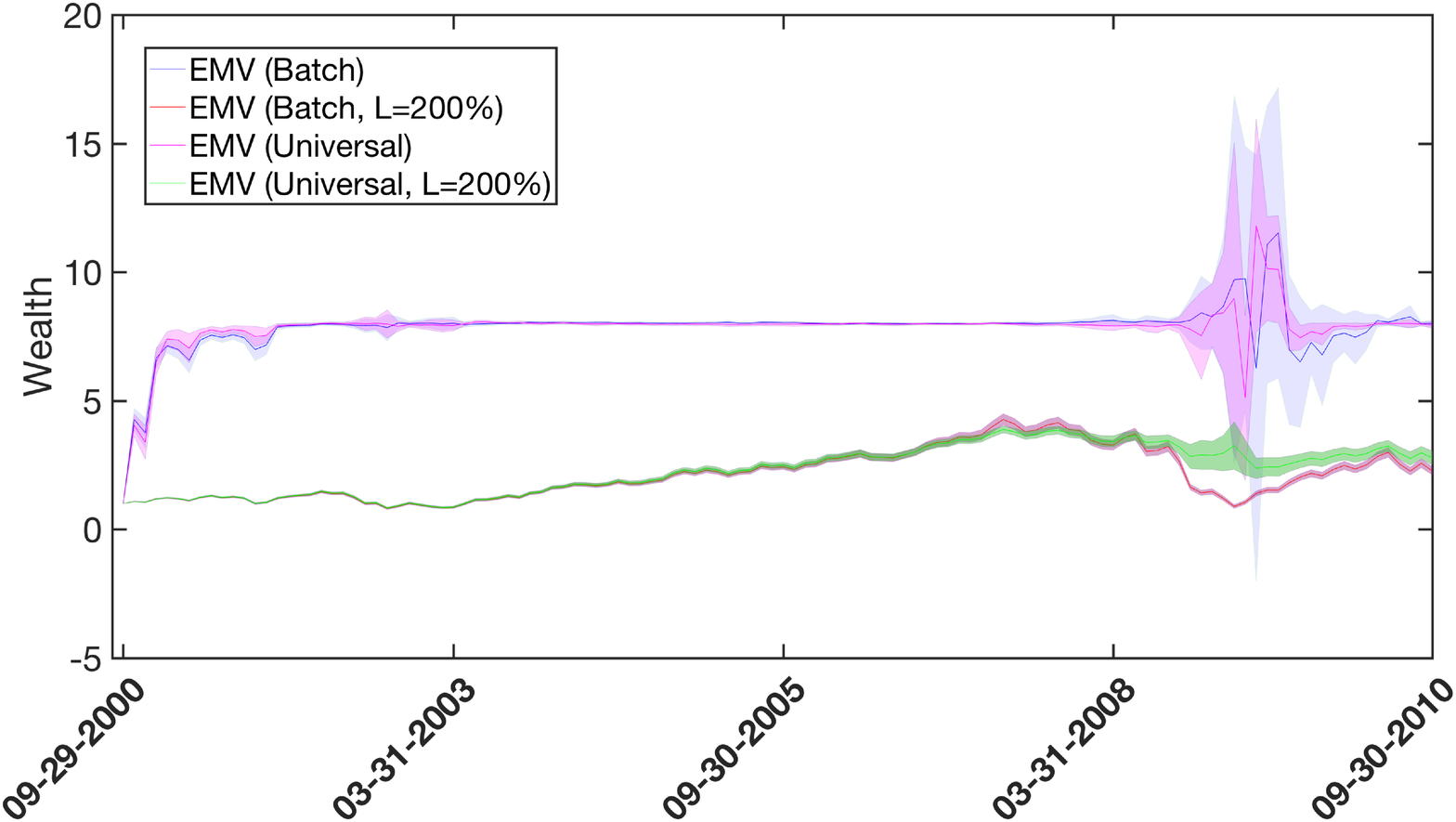}
    \caption{Batch and universal methods}
    \label{fig_2}
  \end{subfigure}
  \caption{Investment performance comparison for (a) batch RL training and testing and (b) batch method and universal method ($d=20$).}
\end{figure}
\end{document}